\documentclass[letterpaper]{article}

\ifdefined\aaaianonymous
    \usepackage[submission]{aaai2026}
\else
    \usepackage{aaai2026}
\fi

\usepackage{times}
\usepackage{helvet}
\usepackage{courier}
\usepackage[hyphens]{url}
\usepackage{graphicx}
\usepackage{natbib}
\usepackage{caption}
\usepackage{amsmath}
\usepackage{amssymb}
\usepackage{amsthm}

\usepackage{mathtools}
\usepackage{booktabs}
\usepackage{array}
\usepackage{placeins}
\usepackage{float}
\usepackage{multirow}
\usepackage{enumitem}
\usepackage{tikz}
\newtheorem{theorem}{Theorem}
\newtheorem{corollary}{Corollary}
\newtheorem{proposition}{Proposition}

\theoremstyle{definition}
\newtheorem{definition}{Definition}

\newtheorem{remark}{Remark}

\newcommand{\Lstar}{L^{\star}}
\newcommand{\spreadK}{\mathrm{spread}_K}
\newcommand{\spreadKone}{\mathrm{spread}@K{=}1}
\newcommand{\rec}{\mathrm{rec}}

\newcommand{\clean}{\mathsf{clean}}
\newcommand{\rand}{\mathsf{rand}}
\newcommand{\zero}{\mathsf{zero}}
\newcommand{\cross}{\mathsf{cross}}

\title{Auditing CoT Answer-Hijack Patches:\
Source-Control Certificates with Type-I Guarantees}

\author{Jianwei Tai}
\affiliations{School of Internet, Anhui University\\
24012@ahu.edu.cn}

\begin{document}

\maketitle
\emergencystretch=3em
\setlength{\hfuzz}{2pt}

\begin{abstract}
Chain-of-thought (CoT) answer-hijack templates can flip the final numeric
answer of a 7B--8B language model on GSM8K or MATH-500 even when the visible
reasoning trace looks fluent. Activation patching is the standard probe for
locating where this hijack can be undone, and a successful clean-source patch
is often read as evidence that the patched activation carries the recovered
content. We show that this reading is unsound: clean-only localization
profiles (peak, spread, thresholded band) underidentify the frozen-hook
source contrast, and the clean-only profile is an intervention map, not a
mediation certificate. We then construct an audit that turns each candidate
patch into a source-control certificate with a pre-registered Type-I
guarantee. The certificate runs in three stages: SELECT (clean-source band
sweep with permutation calibration and held-out validation), FREEZE (lock
the hook), and AUDIT (paired-bootstrap source contrasts at the frozen hook).
It emits an incorrect mechanism label with probability at most
$\alpha=\alpha_{\mathrm{sel}}+\alpha_{\mathrm{audit}}$ under sample-split
disjointness. A matching-rate sample-complexity theorem
($n_\star=\Theta(\Delta^{-2}\log(1/\alpha))$) bounds the audit cost. On
Qwen2.5-7B and Llama3-8B, three few-shot/puzzle cells pass confirmatory
$K{=}1$ localization with held-out gaps $+32.6$, $+45.1$, $+17.7$, fixed-hook
reruns recover $47.0\%$ (Qwen-puzzle) and $39.0\%$ (Llama3-puzzle) at
$n{=}100$, and frozen MATH-500 transfer recovers $26.0\%$. After audit,
Llama3-PZ and Qwen-PZ are identity-light with moderate magnitude (Qwen-PZ
also layer-sensitive), Llama3-FS is a single-seed moderate-positive candidate
(multi-seed replication queued), and Qwen-FS is exploratory non-separation
with a layer-sensitive flag. The method is a diagnostic auditing protocol,
not an adaptive safety defense.
\end{abstract}

\section{Introduction}
\label{sec:intro}

Vision-language and language-only models now solve a growing share of
arithmetic and word-problem benchmarks via chain-of-thought (CoT)
prompting~\cite{wei2022cot,cobbe2021gsm8k}, but the visible chain is not a
faithful audit of the underlying computation: biased exemplars induce
plausible but unfaithful rationales~\cite{turpin2023unfaithful}, chain
perturbations often leave final answers
unchanged~\cite{lanham2023faithfulness}, and intermediate tokens should not be
read as reasoning traces~\cite{gundawar2025stop}. Chain-of-Thought Hijacking
shows that benign-looking reasoning wrappers can redirect a model's final
answer~\cite{zhao2025cothijacking}. Together these results push interpretability
toward \emph{interventions on hidden state} rather than on visible text.

Activation patching is the standard such intervention. A clean run and a
hijacked run are aligned, and the hidden state at one layer of the hijacked
run is replaced by its clean-run counterpart for the first $K$ generated
tokens. If the answer flips back to the gold label, the patched layer is
reported as the locus of the hijack. We argue that this report is routinely
overread. A clean-source recovery curve says that one layer is a useful place
to intervene; it does not say whether the same-problem clean activation is
uniquely responsible for recovery, or whether a random Gaussian, a zero
vector, or an activation drawn from a different layer would have worked just
as well. When auditors care about \emph{which} hidden-state property carries
the recovered content (identity, magnitude, or layer position), a
clean-only patch is silent.

This silence is consequential for safety-adjacent interpretability. CoT
answer hijacks are a controlled proxy for prompt-level redirection of model
output: they do not measure harmful-query safety, but they isolate a
recoverable answer-trajectory failure in which hidden-state interventions can
be scored by deterministic numeric recovery. If patch-based interpretability
mislabels such cells (calling fragility ``mediation'', or calling a layer
``causally important'' on the basis of clean-source recovery alone), any
downstream use of the report (auditing a CoT-instruction-tuned policy, fixing
a brittle reasoning circuit, certifying that a hidden-state edit ``fixes'' a
hijack) inherits the same overinterpretation. We therefore formalize the gap
and turn auditing the gap into a pre-registered procedure with a frequentist
guarantee.

Concretely, our paper separates two claims that clean-source patching
routinely entangles: \emph{where} a candidate intervention reduces hijack
incidence, and \emph{what} property of the patched activation drives the
reduction. A layer can disrupt a hijack trajectory even when the exact clean
source is not the operative semantic carrier; a localized cell can also show
a clean-source advantage over carefully matched controls. The first axis is
about intervention utility, the second is about source identity. Conflating
them produces the standard ``localization implies mediation'' overclaim. We
make the gap quantitative by treating the random-source contrast at the
frozen hook as an estimand whose identifiability under clean-only data is
provably impossible, and we make the audit operational by issuing a
source-control certificate at every cell where a clean-source patch is
proposed.

\paragraph{Contributions.}
\begin{enumerate}[leftmargin=*,itemsep=1pt,topsep=2pt]
\item \textbf{Identifiability gap with matching-rate sample complexity.}
Proposition~\ref{prop:nonidentifiability} gives a clean-only worst-case error
$\ge\tfrac{1}{2}$; Theorem~\ref{thm:sample-complexity} closes the gap with a
Bretagnolle-Huber Bayes-risk lower bound and a Hoeffding upper bound matching
in rate at $n_\star=\Theta(\Delta^{-2}\log(1/\alpha))$ for
$\bar p\in[\eta,1{-}\eta]$, $\eta'=\eta-\Delta>0$;
Proposition~\ref{prop:per-axis} shows each of $\rand,\zero,\cross$ is
individually necessary for its Definition~\ref{def:typology} axis.
\item \textbf{SELECT--FREEZE--AUDIT certificate with Type-I control.}
Definition~\ref{def:certificate} + Proposition~\ref{prop:type-one}: a
pre-registered pipeline with $\alpha=\alpha_{\mathrm{sel}}+\alpha_{\mathrm{audit}}$
emits an incorrect Definition~\ref{def:typology} label with probability
$\le\alpha$.
\item \textbf{Cross-architecture empirical instantiation.} Three
few-shot/puzzle cells pass confirmatory $K{=}1$ localization; $n{=}100$
Llama3-PZ identity-light with moderate magnitude (3 seeds), Qwen-PZ
identity-light with moderate magnitude and layer-sensitive (1 seed); the
single positive clean-source advantage cell, Llama3-FS, is a
\emph{single-seed candidate} (moderate positive, moderate magnitude,
layer-sensitive; multi-seed replication queued); Qwen-FS exploratory
non-separation with layer-sensitive flag (3 seeds).
\end{enumerate}

\section{Related Work}
\label{sec:related}
\paragraph{Activation patching and causal mediation.}
Causal mediation analysis in NLP uses activation interventions to ask which internal states causally support an observed behavior. Vig et al.~\cite{vig2020causal} localized gender-bias effects in GPT-2; ROME and MEMIT locate and edit factual associations in mid-layer MLPs~\cite{meng2022rome,meng2023memit}; IOI circuit work and ACDC extend patching to multi-step tasks and automated circuit discovery~\cite{wang2023ioi,conmy2023acdc}; Patchscopes systematizes patch-based inspection across tasks and models~\cite{ghandeharioun2024patchscopes}. Interchange-intervention and causal-abstraction frameworks~\cite{vig2020causal,hase2023localization} formalize patching as a single-condition counterfactual; we do not claim that this line generally equates causal effect with semantic interpretation. The target here is narrower: a clean-only recovery curve can be read as if it licensed a same-problem clean mediator. Our paper separates the localization question from that inference. We prove that clean-source localization alone cannot identify the clean-vs-random source contrast, and the source-control certificate is the minimal multi-condition extension that resolves the typology in Definition~\ref{def:typology}.

\paragraph{Chain-of-thought failures.}
Chain-of-thought prompting improves arithmetic and commonsense reasoning~\cite{wei2022cot}, with GSM8K as a standard benchmark~\cite{cobbe2021gsm8k}. The visible chain, however, need not be a faithful record of the computation: biased few-shot exemplars can induce plausible but unfaithful rationales~\cite{turpin2023unfaithful}, and chain perturbations often leave final answers unchanged~\cite{lanham2023faithfulness}. Sycophancy is a related failure in which models conform to user-stated beliefs over their own answers~\cite{sharma2024sycophancy}. Chain-of-Thought Hijacking~\cite{zhao2025cothijacking} motivates our setting, but we do not benchmark harmful-query safety. We study a controlled numeric proxy in which few-shot, puzzle, and sycophant wrappers cause GSM8K or MATH-500 final answers to disagree with gold labels, allowing hidden-state interventions to be scored by deterministic numeric recovery.

\paragraph{Activation steering and intervention-based safety.}
Activation steering methods add learned or contrastive directions at inference time, including ActAdd~\cite{turner2023actadd}, Inference-Time Intervention~\cite{li2023iti}, and Representation Engineering~\cite{zou2023repe}. Safety-oriented variants identify toxic-concept, refusal, or safety-vector subspaces and steer along them; Circuit Breakers instead modifies representations during fine-tuning~\cite{zou2024circuitbreakers}. Our intervention is different. It does not learn a population-level direction and does not claim an adaptive defense. It performs a layer-local hidden-state replacement for the first $K$ generated tokens and asks which source-control label is licensed by the resulting recovery: detected clean-source advantage, no detected large clean advantage, or no short-patch localization.

\paragraph{Cross-model localization and negative evidence.}
Some circuits recur across architectures, such as induction heads and IOI-style components~\cite{olsson2022induction,wang2023ioi}, but mechanistic universality is incomplete across model families and architectures~\cite{tan2024universality}. Recent critiques also warn that localized causal importance need not equal stored semantic content~\cite{hase2023localization}, and that intermediate tokens should not be anthropomorphized as faithful reasoning traces~\cite{gundawar2025stop}. Our results fit this skeptical line while making it operational: few-shot and puzzle hijacks expose short-patch bands in Qwen2.5-7B and Llama3-8B, sycophant hijacks do not localize at $K{=}1$, and matched source controls show that some localized cells lack a detected clean-source advantage while another shows a positive clean-source advantage under the tested controls.

\paragraph{Layer-aware robustness editing.}
The closest robustness-adjacent comparison is layer-aware code-LLM editing~\cite{lin2025layeraware}, which changes selected transformer layers in CodeGen and CodeLlama to improve robustness against code perturbations. Both lines identify mid-late layers as privileged intervention sites, but the causal object differs: weight editing improves a model under a fixed perturbation benchmark, whereas our inference-time patches test whether recovery from active CoT answer hijacks should be interpreted as mediation. The distinction matters because a layer can be useful for recovery while failing to identify a large same-problem clean-source advantage.

\begin{table}[t]
\caption{Positioning against representative activation-patching frameworks.
``freq.'' denotes whether a pre-registered Type-I bound on a \emph{mechanism
label} is established (vs.\ metric-level empirical CIs, which several of these
frameworks do report).}
\label{tab:positioning}
\centering
\scriptsize
\setlength{\tabcolsep}{3pt}
\begin{tabular}{@{}lcccc@{}}
\toprule
Framework & src. & freq.\ on label & minimax & port. \\
\midrule
Causal abstr. & multi & --- & --- & informal \\
DAS & multi (learned) & --- & --- & --- \\
Patchscopes & multi (template) & --- & --- & --- \\
ROME/MEMIT & one & --- & --- & --- \\
\textbf{Ours} & 3 axes & Type-I $\le\alpha$ & $\Theta(\Delta^{-2}\log\!\tfrac1\alpha)$ & Rem.~\ref{rem:transfer} \\
\bottomrule
\end{tabular}
\par\vspace{1pt}\footnotesize References:
\citet{geiger2021abstractions,wu2024dasdistributed,ghandeharioun2024patchscopes,meng2022rome,meng2023memit}.
\end{table}

\section{Diagnostic Object}
\label{sec:diagnostic-object}

\subsection{K-Shot Patching Setup}

Let $\pi$ be a transformer language model with $L$ layers. For a problem
$q$ with gold numeric answer $a(q)$, the clean prompt produces a response
$y_c \sim \pi(\cdot\mid q)$. A hijack template $T$ produces a wrapped prompt
$q'=T(q)$ and a hijacked response $y_h \sim \pi(\cdot\mid q')$. We evaluate
only qualified problems for which the clean answer is numerically correct and
the hijacked answer is numerically wrong. This qualification step makes
recovery a binary answer-level quantity rather than a formatting metric.

At generation step $t$, let $h_t^{(\ell)}(q')$ be the hidden state of the
hijacked run at layer $\ell$. A $K$-shot patch replaces the first $K$ generated
hidden states at one layer by a source activation:
\begin{equation}
\hat{h}^{(\ell)}_t =
\begin{cases}
\tilde{h}^{(\ell)}_t & \text{if } t \le K, \\
h^{(\ell)}_t(q') & \text{otherwise.}
\end{cases}
\label{eq:kshot-patch}
\end{equation}
The source $\tilde{h}$ can be the same-problem clean run, a matched random
Gaussian source, a zero source, a cross-layer source, a cross-problem clean
source, or a universal donor activation. The standard mediation reading uses
same-problem clean source by default. Our formal and empirical analyses keep
the source variable explicit.

For source $s$, layer $\ell$, and patch length $K$, define the recovery
indicator and recovery rate as
\begin{align}
R_i(\ell,K,s)&=\mathbf{1}\{A(\hat{y}_{i}^{\ell,K,s})=a(q_i)\},
\label{eq:recovery-indicator}\\
p_s(\ell,K)&=\mathbb{E}_i[R_i(\ell,K,s)],
\label{eq:recovery-rate}
\end{align}
where $A$ is deterministic numeric answer extraction. Clean-source layer
spread at fixed $K$ is
\begin{equation}
\spreadK = \max_{\ell}p_{\clean}(\ell,K)-\min_{\ell}p_{\clean}(\ell,K),
\label{eq:spread}
\end{equation}
and $\Lstar=\arg\max_{\ell}p_{\clean}(\ell,K)$ is a representative peak.
Because $\Lstar$ can be selection-sensitive, we treat thresholded bands as the
stable estimand:
\begin{equation}
\mathcal{B}_{\clean}(K;\tau)=\{\ell:p_{\clean}(\ell,K)\ge \tau\}.
\label{eq:clean-band}
\end{equation}

\subsection{Identifiability Gap}

The missing estimand is the post-selection source contrast. Clean-source
localization observes only one source condition; mechanism interpretation also
requires asking whether the frozen hook recovers under random, zero, cross-layer,
cross-problem, or donor sources. At a selected layer or band, define the source contrast
\begin{equation}
D_s(\ell,K)=p_{\clean}(\ell,K)-p_s(\ell,K).
\label{eq:source-contrast}
\end{equation}
Large positive $D_{\rand}$ or $D_{\zero}$ is a necessary-condition signal for
same-problem source advantage under the tested controls; it is not by itself a
proof of semantic mediation. We use a practical margin $\delta$ to make
identity-light labels falsifiable, and report margin sensitivity rather than
treating one threshold as intrinsic to the task.

\begin{definition}[Mechanism-label typology, multi-axis]
\label{def:typology}
Fix a frozen hook $(\ell_g,K)$, margin $\delta>0$, and target AUDIT-stage
Type-I level $\alpha_{\mathrm{audit}}\in(0,1/2)$. Let
$\overline{D}_{s},\underline{D}_{s}$ denote the
level-$(1-\alpha_{\mathrm{audit}}/3)$ \emph{paired-bootstrap percentile}
upper/lower bounds on
$D_s(\ell_g,K)$ for $s\in\{\rand,\zero,\cross\}$ (single binomial recovery
rates use Wilson). A cell with held-out band gap
$\widehat{\Delta}_{\mathrm{ho}}>0$ whose Wilson 95\% CI excludes $0$ is labeled
by three symmetric axes:
\begin{enumerate}[leftmargin=*,itemsep=2pt,topsep=2pt]
\item \emph{Identity axis} (priority: identity-light $\succ$ strong $\succ$
moderate $\succ$ wide-CI). \emph{Identity-light} if $\overline{D}_{\rand}\le\delta$
(boundary $\overline{D}_{\rand}=\delta$ counts as identity-light);
\emph{strong positive advantage} if $\underline{D}_{\rand}\ge\delta$
(boundary $\underline{D}_{\rand}=\delta$ counts as strong);
\emph{moderate positive advantage} if $\underline{D}_{\rand}>0$ and the CI
brackets $\delta$; \emph{wide-CI} otherwise.
\item \emph{Magnitude flag.} \emph{Strong magnitude} if $\underline{D}_{\zero}\ge\delta$
(boundary $\underline{D}_{\zero}=\delta$ counts as strong); otherwise
\emph{moderate magnitude} if $\underline{D}_{\zero}>0$ (covering both
CI-brackets-$\delta$ and CI-strictly-below-$\delta$ regimes); no flag if
$\underline{D}_{\zero}\le0$.
\item \emph{Layer flag.} \emph{Layer-sensitive} if $\underline{D}_{\cross}>0$.
\end{enumerate}
\emph{Exploratory non-separation} overrides the identity axis only when
$\widehat{\Delta}_{\mathrm{ho}}\le0$, the held-out gap CI crosses zero, the
identity CI is too wide to compute at level $1-\alpha/3$, or sample size is
insufficient; magnitude and layer flags remain reportable when their CI
conditions are met. Each cell inherits one identity label plus optional
magnitude/layer flags (e.g.\ \emph{identity-light, moderate magnitude}).
\end{definition}

\begin{proposition}[Clean-only equivalence class and identifiability floor]
\label{prop:nonidentifiability}
Fix a patch length $K$ and a finite layer set $\mathcal{L}$. For any
clean-source recovery profile $p_{\clean}\in[0,1]^{|\mathcal{L}|}$, define the
\emph{clean-only equivalence class}
\begin{equation}
\resizebox{0.92\linewidth}{!}{$\mathcal{E}(p_{\clean})=\big\{\mathcal{M}:\mathbb{P}_{\mathcal{M}}\{R(\ell,K,\clean)=1\}=p_{\clean}(\ell,K)~\forall\ell\big\}.$}
\label{eq:equivalence-class}
\end{equation}
Let $g$ be any clean-only localization rule, $\ell_g\in g(p_{\clean})$ with
$p_{\clean}(\ell_g,K)>0$, and $\delta\in(0,p_{\clean}(\ell_g,K))$. For any
clean-only decision rule $\phi$ that maps $p_{\clean}$ to the binary hypothesis
$\{D_{\rand}(\ell_g,K)\le\delta\}$ vs.\ $\{D_{\rand}(\ell_g,K)>\delta\}$,
\begin{equation}
\sup_{\mathcal{M}\in\mathcal{E}(p_{\clean})}
\mathbb{P}_{\mathcal{M}}[\,\phi\text{ errs}\,]\;=\;1
\quad(\phi\text{ deterministic}),
\label{eq:det-floor}
\end{equation}
and the randomized minimax error satisfies
$\inf_{\phi}\sup_{\mathcal{M}\in\mathcal{E}(p_{\clean})}\mathbb{P}_{\mathcal{M}}[\phi\text{ errs}]\ge\tfrac{1}{2}$.
\end{proposition}

\begin{proof}
For $\mathcal{M}\in\mathcal{E}(p_{\clean})$, the clean-only marginals are fixed
by~\eqref{eq:equivalence-class}, so $\phi$ is constant on $\mathcal{E}(p_{\clean})$.
Construct $\mathcal{M}_A,\mathcal{M}_B\in\mathcal{E}(p_{\clean})$ with binary
potential outcomes:
$\mathbb{P}_{\mathcal{M}_A}\{R(\ell_g,K,\rand)=1\}=p_{\clean}(\ell_g,K)-\delta'$
with $\delta'\in(\delta,p_{\clean}(\ell_g,K)]$, and
$\mathbb{P}_{\mathcal{M}_B}\{R(\ell_g,K,\rand)=1\}=p_{\clean}(\ell_g,K)$; other
non-clean entries are arbitrary in $[0,1]$. Then $\mathcal{M}_A$ and
$\mathcal{M}_B$ lie on opposite sides of the $\delta$-margin hypothesis but
produce identical $\phi$-input, so any deterministic $\phi$ misclassifies one
of $\{\mathcal{M}_A,\mathcal{M}_B\}$ deterministically (sup-error $=1$). The
randomized bound follows because every randomized rule is a mixture of
deterministic rules.
\end{proof}

\begin{corollary}[At least one non-clean source is necessary]
\label{cor:source-controls}
For any margin $\delta>0$, $\mathcal{E}(p_{\clean})$ contains both
identity-light and positive-advantage models, so identifying any of
labels~1--3 of Definition~\ref{def:typology} requires at least one non-clean
source intervention at the frozen hook.
\end{corollary}

\begin{proposition}[Per-axis necessity]
\label{prop:per-axis}
In the binary potential-outcome model, the joint of
$(R_{\clean},R_{\rand},R_{\zero},R_{\cross})$ at $(\ell_g,K)$ admits any choice
of four source marginals subject to Fr\'echet--Hoeffding compatibility (any
three marginals fixed leave the fourth free in a non-degenerate interval, by
Boole--Fr\'echet bounds on binary marginals), so any three marginals can be
fixed while the fourth varies freely. Hence for each axis
$a\in\{\mathrm{identity},\mathrm{magnitude},\mathrm{layer}\}$ paired with
diagnostic source $s(a)\in\{\rand,\zero,\cross\}$, dropping $s(a)$ leaves axis
$a$ unidentifiable: applying Proposition~\ref{prop:nonidentifiability} to
$D_{s(a)}$ while fixing the other two source marginals constructs
$\mathcal{M}_A,\mathcal{M}_B\in\mathcal{E}(p_{\clean})$ with opposite axis-$a$
labels. Each source is therefore individually necessary.
\end{proposition}

\begin{remark}[Cross-cohort book-keeping]
\label{rem:transfer}
If certificates at level $\alpha$ are issued independently on cohorts $C_0,C_1$,
the joint event ``both certificates correct'' holds at level $\le 2\alpha$ by
union bound. Under exchangeability of the $(\clean,\rand,\zero,\cross)$ joint at
$(\ell_g,K)$ across $C_0\cup C_1$, the $C_1$ label agrees with $C_0$ in
distribution; non-trivial transfer requires re-evaluating Wilson coverage on
$C_1$ rather than reusing $C_0$.
\end{remark}

\begin{theorem}[Matching-rate sample complexity for source-contrast identification]
\label{thm:sample-complexity}
Fix $(\ell_g,K)$, $\delta>0$, $\alpha\in(0,1/4)$, separation $\Delta>0$. Let
$\bar p{=}p_{\clean}(\ell_g,K)$ and let $q_0{=}\bar p-\delta$ denote the
random-source recovery rate at the $H_0/H_1$ boundary; assume
$q_0\pm\Delta\in[\eta',1{-}\eta']$ for some $\eta'>0$. For
$H_0\!:D_{\rand}\le\delta-\Delta$ vs.\ $H_1\!:D_{\rand}\ge\delta+\Delta$ audited
by $n$ paired patches at $(\ell_g,K)$, the Bretagnolle-Huber Bayes-risk bound
gives $\inf_\phi\sup_{\mathcal{M}\in\mathcal{E}(p_{\clean})}\mathbb{P}_{\mathcal{M}}[\phi\,\text{errs}]
\ge\tfrac14 e^{-n\,\mathrm{KL}(q_0-\Delta\|q_0+\Delta)}$. The Bernoulli $\chi^2$
inequality gives the global upper bound
$\mathrm{KL}(q_0-\Delta\|q_0+\Delta)\le 4\Delta^2/[(q_0+\Delta)(1{-}q_0-\Delta)]\le 4\Delta^2/[\eta'(1{-}\eta')]$.
Worst-case error $\le\alpha$ thus requires
$n\ge\log(1/(4\alpha))/\mathrm{KL}\ge\eta'(1{-}\eta')\Delta^{-2}\log(1/(4\alpha))/4$;
the Wilson plug-in achieves error $\le\alpha$ at
$n\ge2\Delta^{-2}\log(2/\alpha)$ by Hoeffding on the paired estimator. Thus
$n_\star(\alpha,\Delta)=\Theta(\Delta^{-2}\log(1/\alpha))$ with constants
depending only on $\eta'$ (lower / upper ratio is $O(\eta'(1{-}\eta'))$); the
lower bound is informative for $\alpha<1/4$ (vacuous as $\alpha\to1/4$).
Proposition~\ref{prop:nonidentifiability} is the $\Delta\to0$ limit;
sub-Gaussian recovery scores follow with the binary KL replaced by
$2\Delta^2/\sigma^2$.
\end{theorem}

\begin{proof}[Sketch]
The lower-bound construction chooses two priors
$\mathcal{M}_A,\mathcal{M}_B\in\mathcal{E}(p_{\clean})$ with
$R_{\clean}\perp R_{\rand}$ and $\mathbb{P}[R_{\rand}=1]\in\{q_0-\Delta,q_0+\Delta\}$,
so the per-sample KL between the two paired-Bernoulli joints reduces to
$\mathrm{KL}(q_0-\Delta\|q_0+\Delta)$ for the random-source marginal. The
$\chi^2$ inequality $\mathrm{KL}(p\|q)\le(p-q)^2/[q(1-q)]$ applied to
$p=q_0-\Delta$, $q=q_0+\Delta$ yields the global KL upper bound;
Bretagnolle-Huber gives the $\tfrac14 e^{-n\mathrm{KL}}$ Bayes-risk lower
bound. The paired estimator $\widehat{D}_{\rand}=\overline{R_{\clean}-R_{\rand}}$
has per-sample range $[-1,1]$, so Hoeffding's inequality gives
$\mathbb{P}[|\widehat{D}_{\rand}-D_{\rand}|>\Delta]\le 2\exp(-n\Delta^2/2)$;
solving for $\le\alpha$ yields $n\ge2\Delta^{-2}\log(2/\alpha)$. Positive
within-pair correlation $\rho=\mathrm{Cov}(R_{\clean},R_{\rand})>0$ further
tightens the upper bound through Bernstein's inequality; the lower bound
construction is independent of $\rho$ by design.
\end{proof}

\begin{figure}[t]
\centering
\begin{tikzpicture}[x=1.0cm,y=0.70cm,>=stealth,font=\footnotesize]
\draw[->,thick] (0,0) -- (6.5,0) node[right]{$\Delta$ (pp)};
\draw[->,thick] (0,0) -- (0,5.6);
\node[anchor=south west,font=\scriptsize] at (-0.2,5.65) {$n_\star$ (samples)};
\foreach \x/\v in {1/10,2/15,3/20,4/25,5/30}{%
  \draw (\x,0.06) -- (\x,-0.06) node[below=1pt,font=\scriptsize]{\v};}
\foreach \y/\v in {0/30,1/60,2/120,3/240,4/480}{%
  \draw (0.06,\y) -- (-0.06,\y) node[left=2pt,font=\scriptsize]{\v};}
\draw[dotted,thick,black!55] (0.2,1.74) -- (6.0,1.74);
\node[black!70,anchor=south west,font=\scriptsize] at (0.25,1.78) {$n{=}100$};
\draw[thick,blue] (1,4.11) -- (2,3.04) -- (3,2.27) -- (4,1.66) -- (5,1.18) -- (5.5,0.93);
\node[blue,anchor=south west,font=\scriptsize] at (3.6,2.85) {Hoeffding upper};
\draw[thick,red,dashed] (1,1.51) -- (2,0.34);
\node[red,anchor=west,font=\scriptsize] at (2.15,0.75) {Le Cam lower};
\fill[red] (1.40,1.74) circle (2.5pt);
\node[red,anchor=south,font=\scriptsize] at (1.40,2.30) {Llama3-FS};
\fill[blue] (4.46,1.74) circle (2.5pt);
\node[blue,anchor=south,font=\scriptsize] at (4.46,2.30) {Llama3-PZ};
\end{tikzpicture}
\caption{Sample-complexity scaling for the source-contrast audit at
$\alpha{=}0.15$, $\eta'{=}0.4$. The blue curve is the Hoeffding upper bound
$n_\star\!\le\!2\Delta^{-2}\log(2/\alpha)$; the red dashed curve is the
Bretagnolle--Huber lower bound
$n_\star\!\ge\!\eta'(1{-}\eta')\Delta^{-2}\log(1/(4\alpha))/4$. Both grow as
$\Delta^{-2}$, with constants depending only on $\eta'$. Filled markers show
the AUDIT cohort size $n{=}100$ at the empirical $\Delta=|\delta-\widehat
D_{\rand}|$ for two cells: Llama3-PZ ($\Delta\!\approx\!22$ pp, well above the
upper-bound curve, identity-light is easy) and Llama3-FS ($\Delta\!\approx\!7$
pp, single-seed candidate near the upper-bound regime).}
\label{fig:sample-complexity}
\end{figure}
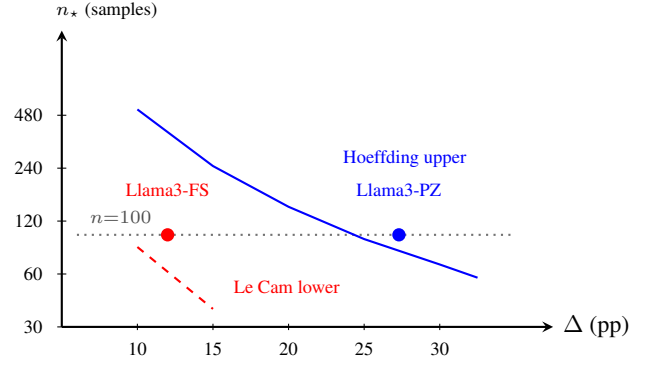

\begin{proposition}[Pre-registered certificate Type-I control]
\label{prop:type-one}
Fix $(\tau,p^{\mathrm{crit}}_{\mathrm{bonf}},\delta,\alpha_{\mathrm{sel}},\alpha_{\mathrm{audit}})$
before any data is observed, with $\alpha=\alpha_{\mathrm{sel}}+\alpha_{\mathrm{audit}}$.
SELECT runs on a cohort $C_{\mathrm{sel}}$ disjoint from the AUDIT cohort
$C_{\mathrm{aud}}$ (assumption A1) and controls $p_{\mathrm{bonf}}\le\alpha_{\mathrm{sel}}$
on the within-cell permutation null over the $K$-grid (Bonferroni over $|K|=5$);
AUDIT then uses paired-bootstrap percentile intervals at level
$1-\alpha_{\mathrm{audit}}/3$ on the three axes of Definition~\ref{def:typology}.
Under A1, total mislabel probability is $\le\alpha$ by union bound on the two
error events. The disjointness in A1 is used to ensure that the AUDIT-stage
paired-bootstrap Wilson coverage is valid as a plug-in (the SELECT-stage
hook-selection randomness does not propagate into the AUDIT cohort); union
bound itself only requires $P(A\cup B)\le P(A)+P(B)$.
\end{proposition}

\begin{remark}[Relation to interchange-intervention frameworks]
\label{rem:interchange}
Causal-abstraction work~\cite{geiger2021abstractions,wu2024dasdistributed}
formalizes interchange interventions, with DAS using distributed multi-source
alignment; causal-mediation analysis of LMs~\cite{vig2020causal} reports
indirect/direct effects; localization critiques~\cite{hase2023localization}
warn that localization need not equal stored content. Our contribution is a
finite-sample audit recipe for the random/zero/cross-layer source contrasts
under Wilson coverage; we instantiate the Hase~\citeyear{hase2023localization}
non-mediation warning in the numeric CoT-hijack setting with a quantitative
worst-case bound.
\end{remark}

\begin{definition}[Source-control certificate]
\label{def:certificate}
A \emph{source-control certificate} for $(\ell_g,K)$ is a tuple
$\big(\mathcal{B},\widehat{\Delta}_{\mathrm{ho}},
\{\widehat{D}_{s}\}_{s\in\{\rand,\zero,\cross\}};\,
\tau,p^{\mathrm{crit}}_{\mathrm{bonf}},\delta,L\big)$ with
$(\tau,p^{\mathrm{crit}}_{\mathrm{bonf}},\delta)$ pre-registered before AUDIT
and entries satisfying: $\mathcal{B}=\widehat{\mathcal{B}}_{\clean}(K;\tau)$
with $p_{\mathrm{bonf}}\le p^{\mathrm{crit}}_{\mathrm{bonf}}$;
$\widehat{\Delta}_{\mathrm{ho}}>0$; paired-bootstrap Wilson 95\% intervals
$\{\widehat{D}_{s}\}$; $L$ from Definition~\ref{def:typology}. SELECT--FREEZE--AUDIT:
clean-source $K$-shot sweep + $5{,}000$ random 50/50 splits (SELECT/Validate);
fix $(\ell_g,K)$ (FREEZE); $5{,}000$ paired bootstraps on
$\rand,\zero,\cross$ at $(\ell_g,K)$ (AUDIT). AUDIT soundness follows from
Theorem~\ref{thm:sample-complexity}; SELECT--FREEZE validity from sample-split
independence.
\end{definition}

\subsection{Selection-Aware Band Validation}

We validate bands on held-out splits: for each 50/50 split the train half selects $\widehat{\mathcal{B}}_{\mathrm{tr}}(K;\tau)=\{\ell:\widehat{p}_{\clean,\mathrm{tr}}(\ell,K)\ge\tau\}$ and the held-out gap is
\begin{equation}
\widehat{\Delta}_{\mathrm{ho}}=
\overline{\widehat{p}_{\clean,\mathrm{ho}}}\big|_{\widehat{\mathcal{B}}_{\mathrm{tr}}}
-\overline{\widehat{p}_{\clean,\mathrm{ho}}}\big|_{\mathcal{L}\setminus\widehat{\mathcal{B}}_{\mathrm{tr}}}.
\label{eq:heldout-gap}
\end{equation}
Confirmatory localization requires Bonferroni-corrected within-cell permutation
evidence and $\widehat{\Delta}_{\mathrm{ho}}>0$; source controls are computed
only after the hook is frozen.

\section{Experimental Protocol}
\label{sec:setup}

\paragraph{Models.}
We evaluate Qwen2.5-7B-Instruct~\cite{qwen2024qwen25} and
Llama3-8B-Instruct~\cite{llama3herd2024}. The main cross-architecture grid
uses INT4 NF4 quantization for single-GPU feasibility. A BF16 Qwen-puzzle
sweep tests whether the main band-level result is a quantization artifact.

\paragraph{Hijack families.}
We use three controlled answer-hijack templates. Few-shot hijacks include
mislabeled exemplars. Puzzle hijacks interleave a misleading riddle with the
user's arithmetic question. Sycophant hijacks prefix the problem with a
confident but incorrect user claim. All templates are instantiated on filtered
benchmark problems and scored only by final numeric answer recovery.

\paragraph{Benchmarks and qualification.}
GSM8K~\cite{cobbe2021gsm8k} is the main benchmark. A problem enters a cell
only when the clean model answer matches the gold label and the hijacked answer
does not. MATH-500 is used for frozen transfer and one Qwen-puzzle source-control
audit. Numeric equivalence handles integer and decimal variants,
comma-separated numerals, simple fractions, degree notation, and final-answer
equations. Qualification yields the diagnostic cohorts used throughout: 236
scanned/187 clean-correct/100 qualified for Qwen-PZ on GSM8K (qualification
rate 42.4\%); 404/303/100 for Llama3-PZ on GSM8K (24.8\%); 500/--/50 for the
Qwen-PZ MATH-500 frozen transfer (10.0\%); and 500/82/50 for the Qwen-PZ
MATH-500 source audit (10.0\%, with 40 cases reused from the fixed transfer
cohort under degree-normalized equivalence and 10 deterministic replacements
drawn by the same rule). All reported recovery estimates are conditional on
clean-correct and hijacked-wrong examples. For each GSM8K puzzle cell, the
SELECT cohort ($n{=}30$ design sweep) and the AUDIT cohort ($n{=}100$ paired
source patches) are drawn under the same qualification rule but use disjoint
problem indices, satisfying the disjointness assumption (A1) of
Proposition~\ref{prop:type-one}.

\paragraph{Layer and patch grid.}
The sweep probes every other layer plus the final layer:
$\{0,2,4,\ldots,L-2,L-1\}$. Patch lengths are
$K\in\{1,2,4,8,16\}$. The canonical short-patch analysis uses $K{=}1$; longer
$K$ values are used to characterize temporal diffusion, especially for
sycophant hijacks.

\paragraph{Statistics.}
Recovery rates use 95\% Wilson; paired source contrasts use $5{,}000$-resample
paired-bootstrap percentile intervals at level $1-\alpha_{\mathrm{audit}}/3$.
Default budgets are $\alpha_{\mathrm{sel}}=0.05$ (within-cell permutation
$p_{\mathrm{bonf}}\le0.05$ over the five $K$ values) and
$\alpha_{\mathrm{audit}}=0.15$, giving total Type-I $\alpha=\alpha_{\mathrm{sel}}+\alpha_{\mathrm{audit}}=0.20$
and per-axis level $1-\alpha_{\mathrm{audit}}/3=0.95$, aligned with the
nominal Wilson 95\% bounds; tighter budgets (e.g.\ $\alpha=0.10$, 98.3\%
bounds) are reported in the supplement. Layer localization uses within-cell
layer-permutation nulls ($5{,}000$ permutations, $p<0.001$ is the floor);
Bonferroni over the five $K$ values defines confirmatory at
$p_{\mathrm{bonf}}\le0.05$, exploratory at $0.05<p_{\mathrm{bonf}}\le0.20$.
Multi-seed random controls resample the seed inside each bootstrap draw.
Selection-aware band intervals use $5{,}000$ random
50/50 splits and summarize conditional diagnostic variability, not independent
replications. Table~\ref{tab:reproducibility} summarizes the settings needed to
replicate the diagnostic protocol.

\begin{table*}[t]
\caption{Reproducibility details for the diagnostic protocol. These settings apply unless a table explicitly states a precision or dataset transfer check.}
\label{tab:reproducibility}
\centering
\scriptsize
\begin{tabular}{p{0.18\linewidth}p{0.74\linewidth}}
\toprule
Component & Setting \\
\midrule
Prompting & Chat template with a math-tutor system instruction; clean prompts ask for step-by-step solving and a final numeric answer; hijack prompts prepend the few-shot, puzzle, or sycophant wrapper. \\
Decoding & Greedy decoding; no sampling; maximum 300 new tokens for GSM8K source-control and fixed-hook runs, 350 new tokens for MATH-500 transfer/source-control runs. \\
Qualification & Include only examples whose clean answer matches the gold numeric answer and whose hijacked answer differs from the clean/gold answer. MATH-500 equivalence also normalizes degree notation before qualification. \\
Intervention & Register a forward hook on the selected transformer block output; replace the last-token hidden state for the first $K$ generated tokens; canonical short-patch setting is $K{=}1$. \\
Sources & Same-problem clean activation, matched-std random Gaussian activation, all-zero activation, cross-layer clean activation, cross-problem clean activation, and universal donor where reported. \\
Layer grid & Every other layer plus final layer for selection sweeps; source controls freeze the hook chosen by the clean-source sweep or fixed-hook replication. \\
Cohorts & Fixed-hook source audits reuse the frozen qualification cohort when it remains qualified under the same decoding and answer-equivalence rule; otherwise deterministic replacements are drawn from the benchmark by the same rule and reported in the qualification table. \\
Statistics & Wilson intervals for recovery rates; within-cell layer-permutation null with 5,000 trials and Bonferroni correction over five $K$ values; paired problem bootstrap with 5,000 resamples for source contrasts; held-out 50/50 split bands with 5,000 random splits for selection-aware validation. \\
Precision & Main grid uses INT4 NF4 quantization; a Qwen-puzzle BF16 sweep checks the band-level precision boundary. \\
Scoring & Deterministic numeric extraction and numeric equivalence; no LLM judge. \\
\bottomrule
\end{tabular}
\end{table*}

\paragraph{Threat and deployment scope.}
The evaluation assumes a white-box operator who can attach layer-level
hidden-state hooks to an open-weights LLM. The operator does not modify model
weights. The attacker controls the text prompt template but is not assumed to
adapt to the disclosed layer and patch length. We therefore evaluate diagnostic
recovery and mechanism interpretation, not adaptive robustness.

\section{Selection-Aware Localization Results}
\label{sec:exp-spread}

Table~\ref{tab:spread} reports the cross-architecture spread surface. The
headline localization result is narrow but stable. At $K{=}1$, Qwen-puzzle,
Llama3-fewshot, and Llama3-puzzle all pass the within-cell permutation test
with Bonferroni-corrected $p_{\mathrm{bonf}}\le0.05$. Their spreads are
0.57, 0.57, and 0.40, respectively. Qwen-fewshot has visible structure but
passes only the exploratory threshold at its best $K$; both sycophant cells
fail the short-patch localization criterion.

\begin{table*}[t]
\caption{Cross-architecture spread surface across 6 (model, hijack)
cells. Each cell is reported at $K{=}1$ (the canonical patch length
used in our recovery checks) and at the per-cell best $K$ from the design
sweep $K \in \{1,2,4,8,16\}$. Layer index $\Lstar$ is reported as
both raw layer and as a fraction of total depth $L$. The
$p_{\text{unc}}$ column is the within-cell layer-permutation null
($5{,}000$ trials per cell, conditional on per-problem propensities);
$p_{\text{bonf}}$ multiplies by 5 to correct for the K-sweep.
Cells with best-K $p_{\text{bonf}} \le 0.05$ are confirmatory localization
signals; $0.05 < p_{\text{bonf}} \le 0.20$ is reported only as an
exploratory small-$n$ signal.}
\label{tab:spread}
\centering
\scriptsize
\begin{tabular}{llcc | ccc | ccccc}
\toprule
& & & & \multicolumn{3}{c|}{at $K{=}1$} & \multicolumn{5}{c}{at best-$K$} \\
Model & Hijack & $L$ & $n$ & $\spreadKone$ & $\Lstar$ & $p_{\text{unc}}$ & $K$ & $\spreadK$ & $\Lstar$ & $p_{\text{unc}}$ & $p_{\text{bonf}}$ \\
\midrule
\multicolumn{12}{l}{\textit{Few-shot and puzzle hijacks}} \\
Qwen2.5-7B   & few-shot  & 28 & 19 & 0.32 & 16 & 0.14   & 4  & \textbf{0.42} & 27 & 0.031 & 0.156 \\
Qwen2.5-7B   & puzzle    & 28 & 30 & 0.57 & 26 & $<\!0.001$ & 1  & 0.57 & 26 & $<\!0.001$ & $<\!0.001$ \\
Llama3-8B    & few-shot  & 32 & 30 & 0.57 & 16 & $<\!0.001$ & 1  & 0.57 & 16 & $<\!0.001$ & $<\!0.001$ \\
Llama3-8B    & puzzle    & 32 & 30 & 0.40 & 14 & $<\!0.001$ & 1  & 0.40 & 14 & $<\!0.001$ & $<\!0.001$ \\
\midrule
\multicolumn{12}{l}{\textit{Sycophant hijacks}} \\
Qwen2.5-7B   & sycophant & 28 & 30 & 0.13 & 10 & 0.94   & 16 & \textbf{0.37} & 27 & 0.020 & 0.101 \\
Llama3-8B    & sycophant & 32 & 30 & 0.07 &  0 & 1.00   & 16 & 0.23 & 10 & 0.097 & 0.485 \\
\bottomrule
\end{tabular}
\par\vspace{4pt}
\footnotesize
At best-$K$ on the original 5-$K$ grid, 3/6 cells reach confirmatory
localization with Bonferroni-corrected $p_{\text{bonf}} \le 0.05$;
two additional cells fall in the exploratory small-$n$ range
$0.05 < p_{\text{bonf}} \le 0.20$. Qwen-sycophant localizes only at
$K{=}16$, consistent with reports that sycophancy emerges over multiple
tokens~\cite{wang2025sycoorigins}; Llama3-sycophant stays below
threshold on this grid; the dense-$K$ exploratory sweep is reported
in the supplement.
\end{table*}

Table~\ref{tab:selection-aware} shows why the paper treats bands rather than
exact layers as the estimand. For each random split, the band is selected on
one half and evaluated on the held-out half. Qwen-puzzle, Llama3-fewshot, and
Llama3-puzzle preserve held-out in-band minus out-of-band recovery gaps of
+32.6, +45.1, and +17.7 percentage points. Exact $\Lstar$ agreement is much
lower: 33.3\% for Qwen-puzzle, 0\% for Llama3-fewshot, and 49.3\% for
Llama3-puzzle. This pattern rules out a strong single-layer story while
supporting a band-level intervention claim.

\begin{table*}[t]
\caption{Selection-aware $K{=}1$ band validation over 5,000 random 50/50 splits. For each split, the thresholded band $\{\ell:\rec_{\mathrm{train}}(\ell,1)\ge .30\}$ is selected on the train half and evaluated on the held-out half. $\Delta$ is held-out in-band minus out-of-band recovery. Exact $\Lstar$ agreement reports how often train and held-out argmax layers match; the split intervals summarize conditional diagnostic variability rather than independent replications.}
\label{tab:selection-aware}
\centering
\scriptsize
\begin{tabular}{lccccc}
\toprule
Cell & $n$ & in-sample $\spreadKone$ & band non-empty & held-out $\Delta$ & exact $\Lstar$ agree \\
\midrule
Qwen-FS & 19 & 0.32 & 92.3\% & +5.5 [-9.2,+20.0] & 51.7\% \\
Qwen-PZ & 30 & 0.57 & 100.0\% & +32.6 [+19.3,+43.3] & 33.3\% \\
Qwen-SY & 30 & 0.13 & 25.1\% & -6.4 [-15.2,-0.5] & 0.0\% \\
Llama-FS & 30 & 0.57 & 100.0\% & +45.1 [+30.2,+60.3] & 0.0\% \\
Llama-PZ & 30 & 0.40 & 89.4\% & +17.7 [+4.0,+34.6] & 49.3\% \\
Llama-SY & 30 & 0.07 & 0.0\% & +0.0 [+0.0,+0.0] & 7.9\% \\
\bottomrule
\end{tabular}
\end{table*}

The negative sycophant result is also informative. Qwen-sycophant forms a
$K{=}1$ train-half band in only 25.1\% of random splits and has a negative
held-out band gap. Llama3-sycophant forms no $K{=}1$ train band at the
$\tau=.30$ threshold. Longer patches partially change this picture:
Qwen-sycophant reaches $\spreadK{=}0.37$ at $K{=}16$ with uncorrected
$p=0.020$ and Bonferroni-corrected $p=0.101$, while Llama3-sycophant remains
below confirmatory threshold on the original grid. We therefore classify
sycophancy as temporal-diffuse under the short-patch protocol rather than as a
failed measurement.

\paragraph{Precision check.}
The BF16 Qwen-puzzle sweep at $K{=}1$ over layers
$\{2,6,10,14,18,20,22,24,26,27\}$ gives recovery rates 13\%, 23\%, 13\%,
27\%, 33\%, 47\%, 40\%, 23\%, 30\%, and 40\%. The spread is 0.33 with a peak
at layer 20. Precision changes the exact peak but preserves a mid/late band,
which supports the band-level claim and weakens the quantization-artifact
objection.

\section{Source-Dependence Controls}
\label{sec:exp-fragility}

Proposition~\ref{prop:nonidentifiability} and Corollary~\ref{cor:source-controls}
turn the missing source axis into an empirical certificate: after a clean-source
sweep selects a useful intervention site, freeze the hook and audit non-clean
sources. The paired controls in Tables~\ref{tab:s7-control},~\ref{tab:paired-source},~\ref{tab:source-certificate},
and~\ref{tab:delta-sensitivity} compare clean-source patching with random, zero,
and cross-layer sources. The practical margin
$\delta=25$ percentage points is pre-registered before any AUDIT-stage data is
observed, together with $(\tau,p^{\mathrm{crit}}_{\mathrm{bonf}},\alpha_{\mathrm{sel}},\alpha_{\mathrm{audit}})$
from Proposition~\ref{prop:type-one}; nearby margins are shown in
Table~\ref{tab:delta-sensitivity} for sensitivity, not for post-hoc selection.
Mechanism labels are certificate-style summaries of paired confidence
intervals, not safety guarantees.

\begin{table}[t]
\caption{Margin sensitivity for the Llama3-puzzle replicated source diagnostic. CR is clean--random using the three-seed random mean; CZ is clean--zero.}
\label{tab:delta-sensitivity}
\centering
\scriptsize
\setlength{\tabcolsep}{3pt}
\begin{tabular}{cccc}
\toprule
$\delta$ & CR upper & CZ upper & Label \\
\midrule
10 pp & +12.0 & +25.0 & unresolved \\
15 pp & below & +25.0 & identity-light \\
20 pp & below & +25.0 & identity-light \\
25 pp & below & touches & mag.-sensitive \\
30 pp & below & below & loose source-invariance \\
\bottomrule
\end{tabular}
\end{table}

\begin{table*}[t]
\caption{How source controls change the interpretation of clean-source localization. The table separates the clean-only intervention claim from the mechanism label obtained after frozen-hook source controls.}
\label{tab:source-certificate}
\centering
\scriptsize
\setlength{\tabcolsep}{3pt}
\begin{tabular}{p{0.14\linewidth}p{0.23\linewidth}p{0.39\linewidth}p{0.16\linewidth}}
\toprule
Cell & Clean-only localization allows & Source controls reveal & Diagnostic label \\
\midrule
Llama3-PZ & Confirmatory localized recovery: held-out band gap +17.7 [+4.0,+34.6]; fixed-hook recovery 39/100. & $n{=}100$: clean 45\%, random 40--45\% (mean 42.3\%), zero 31\%, cross-layer 50\%. Clean$-$random mean +2.7 [-6.7,+12.0]; clean$-$zero +14 [+3,+25]. & identity-light, moderate magnitude \\
\addlinespace
Qwen-PZ & Confirmatory localized recovery: held-out band gap +32.6 [+19.3,+43.3]; fixed-hook recovery 47/100 on GSM8K and 13/50 on MATH-500. & GSM8K $n{=}100$: clean$-$random -2 [-15,+11]; clean$-$zero +20 [+9,+31]; clean$-$cross-layer +15 [+3,+27]. MATH-500 $n{=}50$: clean$-$random +4 [-12,+20]; clean$-$zero +10 [-6,+26]; clean$-$cross-layer +6 [-10,+22]. & identity-light, moderate magnitude, layer-sensitive on GSM8K; MATH magnitude borderline \\
\addlinespace
Llama3-FS & Confirmatory localized recovery: held-out band gap +45.1 [+30.2,+60.3]. & $n{=}100$: clean 70\%, random 52\%, zero 52\%, cross-layer 12\%. Clean$-$random +18 [+8,+28]; clean$-$zero +18 [+8,+28]; clean$-$cross-layer +58 [+47,+68]. & moderate positive clean-source advantage, moderate magnitude, layer-sensitive (single-seed candidate) \\
\addlinespace
Qwen-FS & Exploratory localized recovery only: held-out gap +5.5 [-9.2,+20.0]. & $n{=}22$: clean$-$random +3.0 [-18.2,+27.3]; clean$-$zero +4.5 [-18.2,+22.7]; intervals are too wide for a stronger mechanism label. & exploratory non-separation, layer-sensitive \\
\bottomrule
\end{tabular}
\end{table*}

\begin{table*}[t]
\caption{Paired source-control diagnostics at fixed $\Lstar,K$. Entries are paired recovery differences in percentage points with bootstrap 95\% CIs. Positive values favor clean-source patching under the tested control. Llama-PZ uses the replicated $n{=}100$ audit with a three-seed random mean.}
\label{tab:paired-source}
\centering
\scriptsize
\begin{tabular}{lcccc}
\toprule
Cell & $n$ & clean $-$ random & clean $-$ zero & clean $-$ cross-layer \\
\midrule
Qwen-FS & 22 & +3.0 [-18.2,+27.3] & +4.5 [-18.2,+22.7] & +31.8 [+13.6,+50.0] \\
Llama-FS & 100 & +18.0 [+8.0,+28.0] & +18.0 [+8.0,+28.0] & +58.0 [+47.0,+68.0] \\
Llama-PZ & 100 & +2.7 [-6.7,+12.0] & +14.0 [+3.0,+25.0] & -5.0 [-16.0,+6.0] \\
Qwen-PZ & 100 & -2.0 [-15.0,+11.0] & +20.0 [+9.0,+31.0] & +15.0 [+3.0,+27.0] \\
\bottomrule
\end{tabular}
\par\vspace{2pt}
\footnotesize
Llama-FS shows a positive clean-source advantage under the tested controls at $n{=}100$, although the clean--random gap is smaller than in the $n{=}30$ pilot. Llama-PZ and Qwen-PZ are identity-light against random sources but magnitude- or layer-sensitive under zero and cross-layer controls. Qwen-FS is finite-sample non-separated in the pooled audit and exploratory under held-out band validation.
\end{table*}

\begin{table}[t]
\caption{No-detected-clean-advantage controls at the Qwen-fewshot $K{=}1$ peak ($\Lstar{=}16$). The same numeric-qualified subset is evaluated under three random seeds.}
\label{tab:s7-control}
\centering
\scriptsize
\setlength{\tabcolsep}{3pt}
\begin{tabular}{lccc}
\toprule
Source & Seed-42 & Seed-42 95\% CI & 3-seed range \\
\midrule
Clean       & 10/22 (45.5\%) & [26.9, 65.3] & 45.5--45.5 \\
Random      &  8/22 (36.4\%) & [19.7, 57.0] & 36.4--50.0 \\
Zero        &  9/22 (40.9\%) & [23.3, 61.3] & 40.9--40.9 \\
Cross-layer &  3/22 (13.6\%) & [4.7, 33.3]  & 13.6--13.6 \\
\midrule
No patch    & 0/22 (0\%) & -- & 0 \\
\bottomrule
\end{tabular}
\par\vspace{4pt}
\footnotesize
Same-layer clean, random, and zero interventions do not separate under the reported intervals; the cross-layer source is much weaker. This is consistent with no-detected-clean-advantage layer localization in this cell, but it is not an equivalence test.
\end{table}

\paragraph{Llama3-puzzle: identity-light, moderate magnitude.}
Llama3-puzzle has confirmatory $K{=}1$ localization, a positive held-out band
gap, and fixed-hook large-$n$ recovery 39.0\% on the $n{=}100$ replication
cohort (Tab.~\ref{tab:large-n-checks}). The disjoint AUDIT $n{=}100$ paired
cohort reports clean-source patching at
45\%, random 40--45\% across three seeds (mean 42.3\%), zero 31\%, cross-layer
50\%. Clean$-$random$=+2.7$ [-6.7,+12.0] (below $\delta$, identity-light per
Definition~\ref{def:typology}); clean$-$zero$=+14$ [+3,+25] (lower 3>0, upper
at $\delta$, moderate magnitude); clean$-$cross-layer$=-5$ [-16,+6] (lower
below 0, no layer flag; the negative point estimate is consistent with
neighboring-layer activations sharing residual-stream content, so a different
layer can substitute for the selected hook in this cell).

\begin{table}[t]
\caption{Fixed-hook large-$n$ checks. Hooks are frozen from the GSM8K selection; no layer/K retuning is performed. Wilson 95\% CIs are shown in brackets.}
\label{tab:large-n-checks}
\centering
\scriptsize
\begin{tabular}{lcc}
\toprule
Cell / setting & recovery & Wilson 95\% CI \\
\midrule
Qwen-PZ, GSM8K, $n=100$ & 47.0\% (47/100) & [37.5\%, 56.7\%] \\
Llama-PZ, GSM8K, $n=100$ & 39.0\% (39/100) & [30.0\%, 48.8\%] \\
MATH-500 frozen transfer, $n=50$ & 26.0\% (13/50) & [15.9\%, 39.6\%] \\
\bottomrule
\end{tabular}
\end{table}

\paragraph{Qwen-puzzle: identity-light, moderate magnitude, layer-sensitive.}
GSM8K $n{=}100$ source audit: clean 40\%, random 42\%, zero 20\%, cross-layer
25\%. Clean$-$random$=-2$ [-15,+11] (identity-light), clean$-$zero$=+20$ [+9,+31]
(lower 9>0, upper 31>$\delta$, moderate magnitude), clean$-$cross-layer$=+15$
[+3,+27] (lower 3>0, layer-sensitive). The MATH-500 frozen-hook audit at
$\Lstar{=}26,K{=}1$ ($n{=}50$, 40 from the fixed transfer cohort + 10
deterministic replacements) recovers 28/24/18/22\% under clean/random/zero/
cross-layer; clean$-$random$=+4$ [-12,+20], clean$-$zero$=+10$ [-6,+26],
clean$-$cross-layer$=+6$ [-10,+22]; the random-source non-separation persists
while magnitude becomes borderline.

\paragraph{Llama3-fewshot: moderate positive clean-source advantage, moderate magnitude, layer-sensitive (single-seed candidate).}
In the $n{=}100$ AUDIT (single seed; multi-seed replication queued), clean
source recovers 70\%, random and zero each 52\%, cross-layer 12\%. Paired
clean$-$random and clean$-$zero are both +18 [+8,+28] (lower 8>0, upper
28>$\delta$, moderate positive advantage per Definition~\ref{def:typology});
clean$-$cross-layer is +58 [+47,+68]. The label is conditional on the single
seed (a re-seed could push $\underline{D}_{\rand}$ below 0 and migrate the cell
to wide-CI; see Limits).

\paragraph{Qwen-fewshot: exploratory non-separation, layer-sensitive.}
Qwen-fewshot shows qualitative non-separation but lacks the localization strength
and sample size needed for a stronger mechanism label. In the source-control
audit, the same numeric-qualified 22 problems are evaluated under clean, random,
zero, and cross-layer sources at the frozen $K{=}1$ peak. Clean recovery is
45.5\%, random recovery ranges from 36.4\% to 50.0\% across seeds, and zero
recovery is 40.9\%. The paired clean minus random contrast is +3.0 points
[-18.2,+27.3], and clean minus zero is +4.5 [-18.2,+22.7]. Cross-layer recovery
is much lower, with clean minus cross-layer +31.8 [+13.6,+50.0]; the lower
bound 13.6$>$0 satisfies the layer-sensitive flag of
Definition~\ref{def:typology}, so this flag is reported alongside the
exploratory identity label. The held-out band gap crosses zero and one upper
bound slightly exceeds the 25-point margin, so the identity axis remains
exploratory.

\section{Transfer, Donor, and Baseline Checks}
\label{sec:exp-defense}

\paragraph{Cross-problem and universal-donor recovery.}
Qwen-puzzle provides the strongest donor-reuse check. At the design-sweep hook
$\Lstar{=}26,K{=}1$, self-clean patching recovers 60\% of the original
$n{=}30$ set, while cross-problem clean-source patching recovers 43\%. A
mid-network control at layer 14 recovers only 13\% under self-clean patching
and 20\% under cross-problem patching. A single universal donor activation
recovers 55\% of $n{=}31$ Qwen-puzzle instances, close to the self-clean 61\%
point estimate and above the zero-source 16\% point estimate. These checks do
not prove equivalence between donor and clean source. They show that, in this
cell, problem-specific clean-source identity is not required for substantial
recovery (Tab.~\ref{tab:defense}).

\begin{table}[t]
\caption{Cross-problem source checks for Qwen2.5-7B puzzle at $\Lstar{=}26,K{=}1$ ($n{=}30$). Wilson 95\% CIs in brackets.}
\label{tab:defense}
\centering
\scriptsize
\begin{tabular}{@{}lcc@{}}
\toprule
Patch source layer $\ell$ & Self-clean & Cross-problem \\
\midrule
$\ell = 26$ ($\Lstar$, design sweep) & 60\% (18/30) [42, 76] & \textbf{43\%} (13/30) [27, 61] \\
$\ell = 14$ (mid-network control) & 13\% (4/30) [5, 30] & 20\% (6/30) [9, 38] \\
\bottomrule
\end{tabular}
\par\vspace{2pt}
\footnotesize
\textbf{Universal donor:} one fixed donor activation recovers
$n{=}31$ Qwen-puzzle instances at 55\% (17/31) [38, 71], vs.\ self-clean
61\% (19/31) [44, 76], random 39\% (12/31) [24, 56], and zero 16\%
(5/31) [7, 33].
\end{table}

\paragraph{Fixed-hook large-$n$ checks.}
The original $n{=}30$ design sweep selects hooks and estimates localization.
Table~\ref{tab:large-n-checks} freezes the hook and reruns the two primary
puzzle cells at $n{=}100$. Qwen-puzzle recovers 47.0\% with Wilson 95\% CI
[37.5\%, 56.7\%]. Llama3-puzzle recovers 39.0\% with CI [30.0\%, 48.8\%].
The fixed-hook results lower the point estimates relative to the selection
sweep but preserve nonzero recovery, which is the correct interpretation of a
selection-aware diagnostic.

\paragraph{Frozen dataset transfer.}
MATH-500 transfer freezes the GSM8K-selected Qwen-puzzle hook
$\Lstar{=}26,K{=}1$ with no target-dataset retuning. Recovery is 33.3\% at
$n{=}30$, 30.0\% at $n{=}40$, and 26.0\% at $n{=}50$. A degree-normalized
source-control audit on the same frozen hook recovers 28\% with clean source,
24\% with random source, 18\% with zero source, and 22\% with cross-layer source
at $n{=}50$. The transfer result is bounded but meaningful: it supports nonzero
recovery and preserves the random-source non-separation pattern outside GSM8K
without claiming dataset-universal robustness (Tab.~\ref{tab:math500-transfer}).

\begin{table}[t]
\caption{Frozen MATH-500 transfer check for Qwen-puzzle. The layer and patch length are fixed from GSM8K ($\Lstar{=}26$, $K{=}1$); no MATH-500 retuning is performed. Answers are scored with deterministic numeric-aware equivalence.}
\label{tab:math500-transfer}
\centering
\scriptsize
\begin{tabular}{lcc}
\toprule
Setting & recovery & Wilson 95\% CI \\
\midrule
MATH-500 Qwen-puzzle, $n=30$ & 33.3\% (10/30) & [19.2\%, 51.2\%] \\
MATH-500 Qwen-puzzle, $n=40$ & 30.0\% (12/40) & [18.1\%, 45.4\%] \\
MATH-500 Qwen-puzzle, $n=50$ & 26.0\% (13/50) & [15.9\%, 39.6\%] \\
\bottomrule
\end{tabular}
\end{table}

\paragraph{Baseline interventions.}
On Qwen-puzzle, the paired baseline suite compares patching with input-side and
prompt-side alternatives as intervention-type sanity checks, not as a defense
leaderboard. Patching recovers 60\% on the paired $n{=}30$ cell,
RESTA-style input-embedding smoothing recovers 23\%, same-model paraphrase
sanitization recovers 50\%, and a perplexity-threshold filter recovers 0\%.
A Llama3-puzzle baseline suite preserves the same qualitative pattern: patching
40.0\%, RESTA-style smoothing 10.0\%, paraphrase 33.3\%, and perplexity
filtering 0.0\%. These checks show that the hidden-state intervention is not
just generic input smoothing, but they are not a universal defense ranking.

\section{Discussion and Boundaries}
\label{sec:discussion}

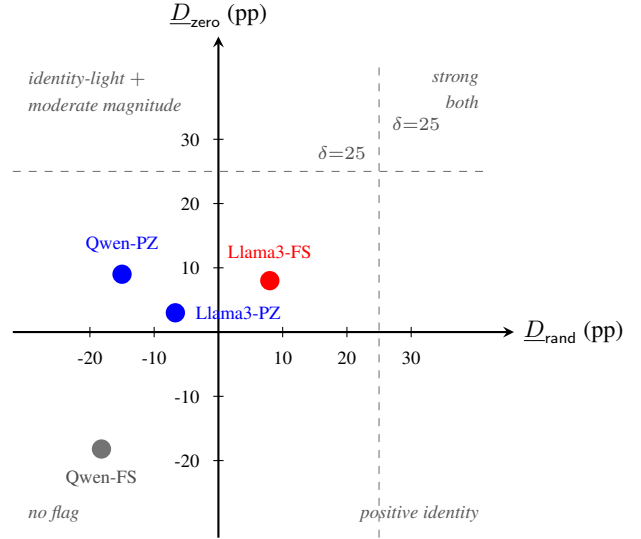
\begin{figure}[t]
\centering
\begin{tikzpicture}[x=0.085cm,y=0.085cm,>=stealth,font=\footnotesize]
\draw[->,thick] (-32,0) -- (46,0) node[right]{$\underline{D}_{\rand}$ (pp)};
\draw[->,thick] (0,-32) -- (0,46) node[above]{$\underline{D}_{\zero}$ (pp)};
\draw[dashed,black!55] (25,-32) -- (25,42);
\draw[dashed,black!55] (-32,25) -- (42,25);
\node[black!70,font=\scriptsize,anchor=south west] at (25.6,30) {$\delta{=}25$};
\node[black!70,font=\scriptsize,anchor=south east] at (24.4,25.5) {$\delta{=}25$};
\foreach \x in {-20,-10,10,20,30}{\draw (\x,0.8) -- (\x,-0.8); \node[anchor=north,font=\scriptsize] at (\x,-1.5) {\x};}
\foreach \y in {-20,-10,10,20,30}{\draw (0.8,\y) -- (-0.8,\y); \node[anchor=east,font=\scriptsize] at (-1.5,\y) {\y};}
\node[font=\scriptsize\itshape,black!65,anchor=north west,align=left] at (-31,42)
  {identity-light $+$\\moderate magnitude};
\node[font=\scriptsize\itshape,black!65,anchor=north east,align=right] at (42,42)
  {strong\\both};
\node[font=\scriptsize\itshape,black!65,anchor=south west] at (-31,-31) {no flag};
\node[font=\scriptsize\itshape,black!65,anchor=south east] at (42,-31) {positive identity};
\fill[blue] (-15,9) circle (1.5);
\node[blue,font=\scriptsize,anchor=south] at (-15,11) {Qwen-PZ};
\fill[blue] (-6.7,3) circle (1.5);
\node[blue,font=\scriptsize,anchor=west] at (-5.0,3) {Llama3-PZ};
\fill[red] (8,8) circle (1.5);
\node[red,font=\scriptsize,anchor=south] at (8,10) {Llama3-FS};
\fill[black!55] (-18.2,-18.2) circle (1.5);
\node[black!70,font=\scriptsize,anchor=north] at (-18.2,-20) {Qwen-FS};
\end{tikzpicture}
\caption{Mechanism-label typology in the
$(\underline{D}_{\rand},\underline{D}_{\zero})$ plane: lower bounds of paired
95\% bootstrap CIs, in percentage points. Dashed lines mark the $\delta=25$ pp
practical margin. The two PZ cells (blue) sit in the identity-light $+$
moderate-magnitude quadrant; the single positive-advantage candidate
Llama3-FS (red) lies in the positive-identity $+$ moderate-magnitude region
but its $\underline{D}_{\rand}=8$ is well below $\delta$, consistent with the
single-seed caveat. Qwen-FS (gray) is in the no-flag quadrant with exploratory
identity status; its $\underline{D}_{\cross}=13.6$ (third axis, not shown)
earns a layer-sensitive flag.}
\label{fig:typology-scatter}
\end{figure}

\paragraph{What the theorem changes.}
Clean-source recovery flags useful intervention sites, not semantic mediators.
Source controls split cases: identity-light with moderate magnitude/layer
sensitivity (Llama3-PZ, Qwen-PZ), moderate positive clean-source advantage
(Llama3-FS), exploratory (Qwen-FS); Qwen-PZ donor checks confirm that
substantial recovery does not always require a matched clean trace. Band
localization is the stable estimand: held-out gaps stay positive while $\Lstar$
agreement ranges 0--49.3\%, and BF16 reproduces the band.

\paragraph{Sycophancy as a temporal-diffuse boundary.}
Sycophant cells mark the short-patch boundary: neither model forms a $K{=}1$
band, and Qwen-sycophant localizes only at longer $K$, consistent with
sycophancy as a multi-stage process~\cite{wang2025sycoorigins,sharma2024sycophancy}.

\paragraph{Limits of the evidence.}
We use two 7B--8B instruction-tuned families on GSM8K with MATH-500 as transfer.
Per-cell random seeds: Llama3-PZ 3, Qwen-PZ 1 (multi-seed planned in Suppl.),
Llama3-FS 1, Qwen-FS 3; cross-cell label differences therefore partly reflect
power, especially the moderate-vs-strong magnitude distinction. Qwen-PZ
replicates the random-source non-separation at $n{=}100$ on GSM8K and at $n{=}50$
on MATH-500 (zero-source borderline); Qwen-FS does not meet the $\delta$
margin. The BF16 precision check is restricted to Qwen-PZ; BF16 sweeps for
Llama3-PZ, Llama3-FS, and Qwen-FS are queued for the arXiv full version. A
token-id-preserving one-token padding stress on a four-problem Qwen-FS smoke
set reduces recovery 50\%$\to$25\%, but random padding does the same; this is
a boundary anecdote, not an optimized-bypass result. Multi-token
gradient-based attacks remain outside scope.

\paragraph{Practical implication.}
Localization claims should state whether they are single-layer or band-level,
how the layer was selected, whether the hook was frozen before reruns, and how
clean source compares with random, zero, cross-layer, and donor controls. A
source-control claim should name its margin $\delta$ and show Wilson upper
bounds on clean-source advantage fall below $\delta$; otherwise the safer label
is identity-light or moderate magnitude. Without these controls, patching
success is intervention utility, not semantic mediation.

\section{Conclusion}
\label{sec:conclusion}

We presented an identifiability + empirical diagnostic study of activation
patching under numeric CoT answer hijacks. Clean-source localization
underidentifies the frozen-hook source contrast (Proposition~\ref{prop:nonidentifiability});
the tight $\Theta(\Delta^{-2}\log(1/\alpha))$ sample complexity
(Theorem~\ref{thm:sample-complexity}) and per-axis necessity
(Proposition~\ref{prop:per-axis}) make a SELECT--FREEZE--AUDIT
certificate at level $\alpha$ feasible from finite samples. Empirically, three
few-shot/puzzle cells pass confirmatory $K{=}1$ localization; source controls
split mechanisms (Llama3-PZ identity-light with moderate magnitude, Qwen-PZ
identity-light with moderate magnitude and layer-sensitive, Llama3-FS moderate
positive clean-source advantage with moderate magnitude and layer-sensitive
(single-seed candidate),
Qwen-FS exploratory, layer-sensitive). Activation patching can identify \emph{where} a hijack is
vulnerable; identifying the source axis requires the certificate.

\begingroup
\sloppy

\section{Extended Tables and Diagnostics}
\label{sec:extended}
The following tables are arXiv-only diagnostics that did not fit the AAAI
9-page main text.

\bibliography{refs}
\endgroup

\end{document}